\documentclass[conference]{IEEEtran}
\usepackage{graphicx,amsmath,amsthm,amssymb,latexsym,epsfig,epstopdf}
\usepackage{comment}
\usepackage{footnote}
\usepackage{color}
\usepackage{lipsum}

\newcommand\blfootnote[1]{%
  \begingroup
  \renewcommand\thefootnote{}\footnote{#1}%
  \addtocounter{footnote}{-1}%
  \endgroup
}

\makeatletter
\def\ps@headings{%
\def\@oddhead{\mbox{}\scriptsize\rightmark \hfil \thepage}%
\def\@evenhead{\scriptsize\thepage \hfil\leftmark\mbox{}}%
\def\@oddfoot{}%
\def\@evenfoot{}}
\def\blfootnote{\xdef\@thefnmark{}\@footnotetext} 

\newtheorem{prop}{Proposition}

\newtheorem{rem}{Remark}

\title{Modelling Load Balancing and Carrier Aggregation in Mobile Networks$^1$}

\author{ F. B\'en\'ezit, S. Elayoubi, R-M. Indre, A. Simonian$^2$,
\\
Orange Labs, Issy-les-Moulineaux, France}

\begin{document}
\maketitle  

\begin{abstract}

In this paper, we study the performance of multicarrier mobile networks. Specifically, we analyze the flow-level performance of two inter-carrier load balancing schemes and the gain engendered by Carrier Aggregation (CA). CA is one of the most important features of HSPA+ and LTE-A networks; it allows devices to be served simultaneously by several carriers. We propose two load balancing schemes, namely Join the Fastest Queue (JFQ) and Volume Balancing (VB),  that allow the traffic of CA and non-CA users to be distributed over the aggregated carriers. We then evaluate the performance of these schemes by means of analytical modeling. We show that the proposed schemes achieve quasi-ideal load balancing. We also investigate the impact of mixing traffic of CA and non-CA users in the same cell and show that performance is practically insensitive to the traffic mix. 

\end{abstract}


\section{Introduction}
\label{sec:intro}


\blfootnote{$^1$ This work has been partially financed by the ANR IDEFIX project.}
\blfootnote{$^2$ The order of listing the authors is alphabetical.}

One of the most important features of High Speed Packet Access (HSPA) and Advanced Long Term Evolution (LTE-A) is {\it Carrier Aggregation} (CA) which allows users to be served simultaneously by two or more carriers. 
The HSPA spectrum is divided into carriers of 5 MHZ each, while the LTE spectrum is divided into carriers with sizes ranging from 1.4 MHZ to 20 MHZ. The classical way of managing these carriers is to consider an independent scheduler per carrier. This type of resource management, however, may result into inefficiencies due to load discrepancies between carriers. To cope with this inefficiency, HSPA+ has defined carrier aggregation. Specifically,  Release 8 introduces the Dual Carrier (DC) feature where  two carriers are aggregated on the frequency band of 2.1 GHZ. Release 10 has extended this concept to the aggregation of a carrier on 2.1 GHZ with another carrier on the 900 MHZ band; this is the Dual Band (DB) HSPA feature in which the two carriers can have significantly different capacities due to the difference of propagation conditions between the 900 MHZ and the 2.1 GHZ bands. In LTE-A, it is possible to aggregate two or more carriers of different sizes, leading to large capacity difference between carriers.

It is widely agreed that carrier aggregation brings gain in the following two aspects. First, the user peak rates is substantially increased (e.g., doubled for DC devices). However, this is only true at low load where the user is almost always alone in the cell. In this paper, we investigate the {\it actual} throughput gain of DC users when the capacity of the carriers is dynamically shared by several mobile users. The second advantage is joint scheduling or {\it pooling}. Indeed, since a single scheduler is used for two or more carriers, the scheduler can implement intelligent load balancing schemes so as to equalize the traffic over the carriers and thus increase the traffic capacity. In the present paper, we propose two load balancing schemes and evaluate their performance in terms of traffic capacity and average flow throughput. 

\subsection{Related work}

Quantifying the gain of carrier aggregation in downlink HSPA (HSDPA) has been the object of several papers. Authors in \cite{Holma} have shown that DC HSDPA Release 8 
doubles the throughput of dual carrier users only at low network load. The authors also show that even at high network load, this feature still yields considerable gain when compared to single carrier HSPA Release 7. In \cite{6}, the gains expected from frequency diversity and higher multi-user diversity have been evaluated; these gains are used as inputs for the capacity model used in the present paper. Simulation results are provided in \cite{7} and \cite{8} for full-buffer traffic, i.e., the number of users in the system is fixed and these users have an infinite amount of data to transmit. However, the trunking gain cannot be observed for full buffer traffic, as explained in \cite{6}. This gain has been assessed in \cite{bon11}; results show that we have a 45\% improvement in DC-HSDPA capacity. The latter could exceed 160\% if Dual Carrier is combined to MIMO (DC HSDPA Release 9).

 Another set of works has focused on carrier aggregation in the uplink. The feasibility of dual carrier in uplink HSPA  (HSUPA) has been studied in \cite{3}. It is shown that, in cells with a relatively large radius, cell edge users can barely benefit from their dual carrier capabilities. The pooling gain for DC HSUPA has been studied in \cite{Amal}, where it has been shown that, due to power limitations at the user device, the gain from carrier aggregation is lower in the uplink than in the downlink. 

In LTE-A, carrier aggregation has been initially designed in order to allow the extension of the offered bandwidth by transmitting over multiple carriers, so that, from a user's point of view, the aggergated carriers are seen as a single, large carrier \cite{Ratasuk}, \cite{Yuan}. System level simulations in \cite{Pedersen} show that the carrier aggregation gain is large at low load and decreases when the number of scheduled users becomes large. The same trend is shown in \cite{Zhang}, also using system level simulation; it is also observed that the  full buffer traffic model yields an over-estimation of the carrier aggregation gain. Paper \cite{Wang} confirms this trend in the framework of inter-band carrier aggregation considering carriers on 800 MHZ and 2.1 GHZ frequency bands.

\subsection{Contribution} 
The contribution of this paper is twofold. First, we propose two inter-carrier load balancing schemes that allow to efficiently split incoming CA and non-CA traffic over the two carriers. We then analyse the performance of the proposed schemes by modelling the cell occupancy by means of a Markov process. Based on this analytical model, we show that the proposed schemes achieve quasi-ideal load balancing over the aggregated carriers. 

As discussed in the previous section, most of the existing work on carrier aggregation uses system level simulation based on the full-buffer traffic model in order to assess the performance gains of carrier aggregation. While the full-buffer model allows to accurately model the lower layers and to estimate the impact of physical effects such as path loss, shadowing and fast fading, it does not represent a realistic traffic model since it does not account for the dynamic behaviour of users. There is thus a need for analytical models that capture the flow-level dynamics and help understand the system performance. This paper fills this gap by proposing analytical models that allow to evaluate the system performance while accounting for the following essential features: 

\begin{itemize}
\item Aggregation of carriers of different capacities, which is the case of LTE-A and DB HSDPA
\item A mix of CA capable devices and non-CA devices. This includes the case of legacy devices that do not support CA coexisting with CA devices and also the case of users that are covered by only one of the carriers (e.g., an indoor user that is covered by the 900 MHZ band only)
\item The interaction of the load balancing feature with the CA. Indeed, legacy users do not select randomly the serving Carrier as it is supposed in \cite{bon03}, but take into account the rate that they will get on both carriers before making a decision.
\end{itemize}




\section{System description}
\label{sec:network}


In this section, we present the considered network architecture and briefly describe the scheduling schemes proposed for {\it Single Carrier} (SC) and {\it Dual Carrier} (DC)  users.
 
\subsection{Radio conditions}
We consider the downlink of a cell equipped with two carriers. The capacity of each carrier is time-shared between active users. There are two types of users in the cell: SC users, which are allowed to use only one of the two carriers and DC users, which are allowed to simultaneously use both carriers. The resources of the two carriers are shared using a single joint scheduler, as explained in the next section.

Let $C_1$ and $C_2$ denote the maximum peak data rates at which a user may transmit on Carrier 1 and Carrier 2 in the vicinity of the base station. This peak bit rate depends on the radio environment, the bandwidth and the coding efficiency. Typically, the peak data rate varies over time due to user mobility, shadowing and multipath reflections. As the aim of this paper is not to model these physical effects, we simply consider that there are $J$ areas in the cell, with area $A_j$ characterized by throughputs $C_{1,j}$ and $C_{2,j}$ corresponding to the maximum rate that can be attained on Carrier 1 and Carrier 2 in area $A_j$ when a single user is present in the cell. Furthermore, we assume that the position of active users remains constant in time and that the probability of being in area $j$ is equal to $q_j$. A classical assumption is to consider a cell divided into concentric rings and that users are uniformly distributed over the cell surface. Let $R$ denote the total cell radius. The probability of being in area $A_j$ is thus given by
\begin{equation}
q_j=\frac{\vert A_j \vert}{\pi R^2}
\label{qj}
\end{equation}
where $\vert A_j \vert = \pi(r_j^2-r_{j-1}^2)$ is the area of ring $j$ ($r_j$ is the radius of ring $j$ and by convention, $r_0 = 0$ and $r_J=R$). 

\subsection{Scheduling schemes}
\label{sec:Sched1}
The scheduler must first decide how to assign SC and DC users to one or both carriers and then share the capacity of each carrier among the different users.


Upon the arrival of a new SC user, the scheduler decides which of the two carriers should serve this user for its entire session; and upon the arrival of an incoming DC flow, the scheduler must decide how much of the volume of that flow should be treated by each carrier, given that DC users are served by both carriers. We assume that this volume allocation can vary during the DC session. 

Once these scheduling decisions are made, the scheduler fairly divides the capacity of each carrier among the active users present on each carrier. Specifically, if $n_1$ users are present on Carrier $1$, each user in area $A_j$ is allocated a fraction $C_{1,j}/n_1$ of the total carrier rate $C_{1,j}$; this corresponds to a \textit{Processor Sharing} (PS) service discipline in which the server capacity is equally shared among active users on that Carrier. Similarly, if $n_2$ users are present on Carrier $2$, each user in area $A_j$ is allocated a fraction $C_{2,j}/n_2$. 

We propose SC and DC scheduling schemes that maximize the immediate rate of each user in order to study a \emph{best-case scenario}.
\subsubsection{Scheduling SC users}
\label{sec:Sched1SC}

We propose to schedule single carrier users via the \textit{Join the Fastest Queue}  (JFQ) scheduling policy, where the incoming SC flow is assigned to the carrier which would provide this flow with the smallest completion time, or equivalently with the largest throughput. Specifically, assume that Carriers 1  and  2  are currently serving $n_1$ and $n_2$ flows, respectively. Under the assumption of fair sharing,  the carrier that provides the smallest completion time is  the one that ensures the largest ratio among $C_{1,j}/(n_1+1)$ and $C_{2,j}/(n_2+1)$. Note that the completion times are computed only on the basis of the state of the system at the arrival time of the SC customer; system state changes (i.e., arrival or departure of other flows) that may occur during the transmission of the considered flow are thus not (and cannot be) taken into account in the scheduling scheme.


\subsubsection{Scheduling DC users}
\label{sec:Sched1DC}

Suppose that, at a given moment, the base station has an amount of data $\sigma_0$ that is relative to a given DC user in its buffer. The scheduler must decide how to split the volume $\sigma_0$ between the two carriers. Let $\sigma_1$ and $\sigma_2$ denote the volumes of that flow respectively transferred by Carrier 1 and Carrier 2 (with $\sigma_1 + \sigma_2 = \sigma_0$), and let $T_1$ and $T_2$ be the associated completion times. Ideally, volume $\sigma_0$ should be split such that the transfer of volumes $\sigma_1$ and $\sigma_2$ is completed simultaneously on both carriers. In fact, the actual completion time of the DC flow is determined by the slowest of the two carriers, i.e., $T = \max(T_1,T_2)$. To minimize $T$, the total volume must be split such that the completion times on the two carriers are identical, i.e., $T_1=T_2$. In the following, we will propose a scheduling policy that minimizes $T$ and ensures simultaneous completion times. It will be referred to as the {\it Volume Balancing} (VB) and will be formally defined in Section \ref{DCarr}. 



In practice, perfect Volume Balancing may be difficult to implement due to rapidly changing rates caused by fast fading. The scheduler should send batches of frames to each carrier, the size of a batch being  \emph{proportional to the current rate of each carrier}. The  idea is that batch sizes adapt to the time-varying rates, and so the carriers finish serving the batches of frames at the same time. Note that this mechanism works if the typical batch service time is smaller than the typical interval during which the rate remains constant. If this condition is not fulfilled, an alternative is to compute a moving average of carrier rates over a sliding window and adapt to this more stable rate value. In order to derive best-case results, we hereafter neglect the above implementation details and study an ideal scheduler, where the data is perfectly split among carriers.



\section{System model}
\label{sec:model}


In this section, the network model is presented. We first specify the proposed scheduling policies, i.e. Join the Fastest Queue and Volume Balancing, in terms of system parameters. We then model the system by means of a Markov process and discuss system stability and throughput performance.

\subsection{Scheduler modeling}
\label{JS}

Given the fair sharing assumption (cf. \S \ref{sec:Sched1}), the system can be modeled by two multiclass PS queues 
representing the two carriers, as depicted in Figure \ref{fig:Mix}. At any time $t$, the state of the system is defined by the 3$J$-dimensional vector 
$$
\mathbf{n}(t) = (n_{1,1}(t),n_{2,1}(t),m_1(t); \ldots ;n_{1,J}(t),n_{2,J}(t),m_J(t) )
$$
where 
\begin{itemize}
\item[-] $n_{1,j}(t)$ 
is the number of SC users in area $A_j$ currently transmitting on Carrier 1,
\item[-] $n_{2,j}(t)$ 
is the number of SC users in area $A_j$ currently transmitting on Carrier 2,
\item[-] $m_j(t)$ 
is the number of DC users in area $A_j$ currently transmitting on Carriers 1 and 2. 
\end{itemize}
The fact that a single variable is sufficient to characterize DC flows in both queues will be justified at the end of this section. For $1 \le j \le J$, let $\mathbf{e}_{1,j}$, $\mathbf{e}_{2,j}$ and $\mathbf{e}_{3,j}$, denote the unit vectors such that $\mathbf{e}_{k,j}$ has the $(3j+k-1)$-th element equal to 1, while all other elements are 0. Finally, define the rate functions $D_{1,j}$ and $D_{2,j}$ by
\begin{equation}
\left\{
\begin{array}{ll}
D_{1,j}(\mathbf{n}(t)) = \displaystyle \frac{C_{1,j}}{n_1(t)+m(t)}, 
\\ \\
D_{2,j}(\mathbf{n}(t)) = \displaystyle \frac{C_{2,j}}{n_2(t)+m(t)}
\end{array} \right.
\label{defD1D2}
\end{equation}
where
$$
n_1(t)=\sum_{j=1}^{J} n_{1,j}(t), \; \; n_2(t)=\sum_{j=1}^{J} n_{2,j}(t), \; \; m(t)=\sum_{j=1}^{J} m_j(t).
$$

\begin{figure}[ht]
\begin{center}
\includegraphics[width=5cm, trim=1cm 3cm 1cm 3cm, clip]{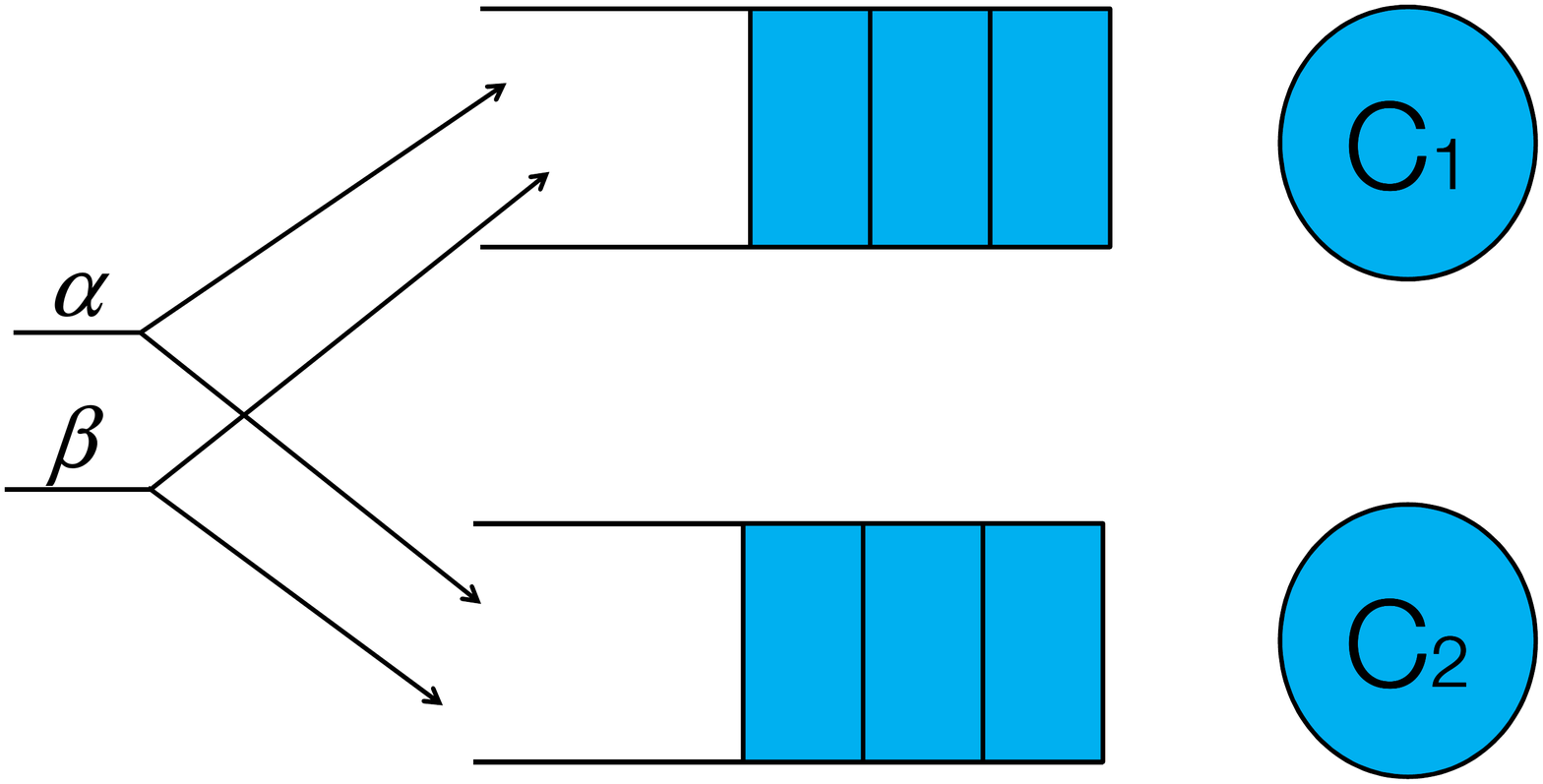}
\caption{\label{fig:Mix} \textit{Queueing model for joint SC and DC flows.}}
\end{center}
\end{figure}

\subsubsection{JFQ policy for SC users}
\label{SCarr}

As explained in Section \ref{sec:Sched1SC}, upon the arrival of a SC user at time $t$ in area $A_j$, the user is directed towards either Carrier 1 queue or Carrier 2 queue according to the JFQ scheduling discipline. Let $t_+$ denote any instant close enough to time $t$ so that no other event occurs in the system during interval $]t,t_+]$; the system state at time $t_+$ is then updated as follows: 
\begin{itemize}
\item[\textbf{\textit{i1)}}] if $\displaystyle \frac{C_{1,j}}{n_1(t)+m(t)+1} > \frac{C_{2,j}}{n_2(t)+m(t)+1}$, \\ \\
the flow is directed to Carrier 1 and $\mathbf{n}(t_+) = \mathbf{n}(t)+\mathbf{e}_{1,j}$;
\\
\item[\textbf{\textit{i2)}}]  if $\displaystyle \frac{C_{2,j}}{n_2(t)+m(t)+1} > \frac{C_{1,j}}{n_1(t)+m(t)+1}$,
\\ \\
the flow is directed to Carrier 2 and $\mathbf{n}(t_+) = \mathbf{n}(t)+\mathbf{e}_{2,j}$;
\\
\item[\textbf{\textit{i3)}}] if $\displaystyle \frac{C_{1,j}}{n_1(t)+m(t)+1} = \frac{C_{2,j}}{n_2(t)+m(t)+1}$, \\ \\
the flow is directed either to Carrier 1 or to Carrier 2 with probability $1/2$.
\\
\end{itemize}

\subsubsection{VB policy for DC users}
\label{DCarr}

DC flows are scheduled according to the VB policy which is defined by the following property:

for any given DC flow in any area $A_j$, and for any time interval $[t,t')$ where the system state $\mathbf{n}(\cdot)$ remains constant, Carrier $1$ and Carrier $2$ transfer volumes $\sigma_1$ and $\sigma_2$, respectively, such that
\begin{equation}
\left\{
\begin{array}{ll} 
\sigma_1 = D_{1,j}(\mathbf{n}(t))(t'-t), \\ \\
\sigma_2 = D_{2,j}(\mathbf{n}(t))(t'-t).
\end{array} \right.
\label{defvolumes}
\end{equation}
In other words, this property ensures that any DC flow exploits the full potential of both its carrier ressources, i.e. the VB property characterizes the best possible scheduler. The feasibility of Volume Balancing was discussed in \ref{sec:Sched1DC}. Volume Balancing ensures the following interesting property:  


\begin{rem} 
The VB policy ensures that the transfer of any DC flow on Carriers 1 and 2 is completed at the same time on both carriers.

\label{PropVB}
\end{rem}

In fact, as soon as a DC flow exists at time $t$, the VB policy ensures that \emph{both} carriers are serving the flow. Therefore, they finish serving simultaneously.

\begin{rem}\label{rem:flo}
In the interval $[t,t')$ defined above, frames are served proportionally to rates:
\begin{equation}
t'-t= \frac{\sigma_1}{D_{1j}(\mathbf{n}(t))} =  \frac{\sigma_2}{D_{2j}(\mathbf{n}(t))} = \frac{\sigma_1 + \sigma_2}{E_j(\mathbf{n}(t))}
 \label{eq:alain}
 \end{equation}
where
 \begin{equation}
E_j(\mathbf{n}(t)) = D_{1j}(\mathbf{n}(t)) + D_{2j}(\mathbf{n}(t)).
\label{Ej}
\end{equation}
\end{rem}
This readily follows from Eq.~\eqref{defvolumes} and simple algebra. The last equality in Eq.~\eqref{eq:alain} illustrates that the system behaves as a single carrier of rate $E_j(\mathbf{n}(t))$ for the entire DC volume $\sigma_1+\sigma_2$.

Given that DC flows arrive and depart simultaneously in both queues, a single state variable can be therefore used to characterize the number of DC flows in the system, as claimed in the beginning of this section.

\subsection{A Markov process}
\label{sec:Markov}

Assume that SC and DC flows arrive in the system according to Poisson processes with respective arrival rates $\alpha$ and $\beta$, and that the volume of data of any (SC or DC) user is exponentially distributed with mean $\sigma$. We denote by $\Lambda=\alpha+\beta$ the total arrival rate in the cell. We further assume that the traffic is uniformly distributed over the cell such that in area $A_j$, SC and DC flows arrive with intensities $\alpha_j=\alpha q_j$ and $\beta_j= \beta q_j$, with probability $q_j$ introduced in (\ref{qj}). 

Based on the assumptions of exponential flow size distribution and Poisson flow arrivals and given conditions \textbf{\textit{i1}}, \textbf{\textit{i2}} and \textbf{\textit{i3}} defined in Section \ref{SCarr}, the system state described by vector $\mathbf{n}(t)$, $t \geq 0$,  defines a Markov process which can move in an infinitesimal time interval from state $\mathbf{n}$ to state:

\begin{itemize}
\item[T1)] $\mathbf{n} + \mathbf{e}_{1j}$ with transition rate
$
\alpha_j \mathbf{1}_{\textbf{\textit{i1}}}
+ \frac{\alpha_j}{2}\mathbf{1}_{\textbf{\textit{i3}}} 
$
\item[T2)] $\mathbf{n} + \mathbf{e}_{2j}$ with transition rate
$
\alpha_j \mathbf{1}_{\textbf{\textit{i2}}}
+ \frac{\alpha_j}{2}\mathbf{1}_{\textbf{\textit{i3}}} 
$
\item[T3)] $\mathbf{n} + \mathbf{e}_{3j}$ with transition rate $\beta_j$
\item[T4)] $\mathbf{n} - \mathbf{e}_{1j}$ with transition rate $n_{1j}D_{1j}(\mathbf{n})/\sigma$, since Carrier 1 behaves as a PS queue with departure rate proportional to the number $n_{1j}$ of SC users in $A_j$
\item[T5)] $\mathbf{n} - \mathbf{e}_{2j}$ with transition rate $n_{2j}D_{2j}(\mathbf{n})/\sigma$,  since Carrier 2 behaves as a PS queue with departure rate proportional to the number $n_{2j}$ of SC users in $A_j$
\item[T6)] $\mathbf{n} - \mathbf{e}_{3j}$ with transition rate $r_j(\mathbf{n)} = m_j E_j(\mathbf{n})/\sigma$, where $E_j(\mathbf{n})$ is defined in Eq.~\eqref{Ej}.

\end{itemize}
To justify the value of $r_j(\mathbf{n})$, recall from Remark \ref{rem:flo} that the completion of a DC flow ends simultaneously on both carriers, which behave as one single carrier of rate $E_j(\mathbf{n})$.

\subsection{System stability}
\label{sec:stab}

We now address the system stability. We start by considering the case in which radio conditions are uniform over the entire cell so that $J=1$ and the capacity of each carrier is simply $C_{1,1}=C_1$ and $C_{2,1}=C_2$.

\begin{prop}
Defining the system load (with single area) by
\begin{equation}
\rho = \frac{(\alpha+\beta)\sigma}{C_1+C_2},
\label{eq:rhoJ1}
\end{equation}
the stability condition 
is $\rho < 1$.
\end{prop}

\begin{proof}
A {\it necessary} condition for stability for a single area cell  is that the total traffic intensity $(\alpha + \beta)\sigma$ must be less than the total cell capacity $C_1 + C_2$, that is, $\rho < 1$ with load $\rho$ defined in (\ref{eq:rhoJ1}). 

To prove that $\rho < 1$ is also a {\it sufficient} stability condition, we use a fluid limit approach  (\cite{robert}, Chap.9). Specifically, we denote by
\begin{itemize}
\item $V_1$, the volume (in Mbit) of SC traffic in queue 1,
\item $V_2$, the volume (in Mbit) of SC traffic in queue 2,
\item $W$, the volume (in Mbit) of DC traffic in both queues.
\end{itemize}
Given the assumptions of exponential flow size and Poisson arrivals, the system state described by vector $(V_1(t),V_2(t),W(t))$ defines a Markov process in $\mathbb{R}^{3+}$. The latter is ergodic if the fluid model is stable in the sense that, starting from any initial state, the total fluid volume $V_1+V_2+W$ reaches a bounded state of $\mathbb{R}^{3+}$ in finite time. Using the transition rates expressed in Section \ref{sec:Markov} with $J = 1$, the evolution of the system can then be described by the equations
\begin{equation}
\left\{
\begin{array}{ll}
\displaystyle \frac{dV_1}{dt} = \displaystyle \alpha\mathbf{1}_{\mathbf{i1}}\sigma+\frac{\alpha}{2}\mathbf{1}_{\mathbf{i3}}\sigma - \frac{C_1 n_1}{n_1+m},
\\ \\
\displaystyle \frac{dV_2}{dt} = \displaystyle \alpha\mathbf{1}_{\mathbf{i2}}\sigma+\frac{\alpha}{2}\mathbf{1}_{\mathbf{i3}}\sigma - \frac{C_2 n_2}{n_2+m},
\\ \\
\displaystyle \frac{dW}{dt} = \displaystyle \beta\sigma - \left ( \frac{C_1}{n_1+m}+\frac{C_2}{n_2+m} \right ) m,
\end{array} \right.
\label{eq:fluid}
\end{equation}
where conditions $\textbf{\textit{i1}}$, $\textbf{\textit{i2}}$ and $\textbf{\textit{i3}}$ are introduced in Section \ref{SCarr}, and $n_1$, $n_2$ denote the number of SC customers in queue 1 queue 2, respectively,  while $m$ is the number of DC customers in both queues. Summing all three equations (\ref{eq:fluid}) side by side, we obtain $ \frac{d}{dt}(V_1+V_2+W) = (\alpha + \beta) \sigma - C_1 - C_2$ 
which $< 0$ as long as  $\rho < 1$. This ensures that the total volume $V_1 + V_2 + W$ is a decreasing function of time, and that the system empties in a finite time. Inequality $\rho < 1$ is therefore also a  sufficient condition for stability.
\end{proof}

Let us now analyze the system stability in the more general case with $J>1$ different areas in the cell, each with its own radio conditions.
\begin{prop}
Define the system capacity $\overline{C}$ (with multiple areas) by the harmonic mean
\begin{equation}
\frac{1}{\overline{C}} =  \sum_{j=1}^J \frac{q_j}{C_{1,j} + C_{2,j}}
\label{eq:cap}
\end{equation}
of total capacities $C_{1,j} + C_{2,j}$ in area $A_j$, weighted by probabilities $q_j$, $1 \leq j \leq J$, introduced in (\ref{qj}). Denoting the system load by
\begin{equation}
\rho = \frac{(\alpha + \beta)\sigma}{\overline{C}},
\label{eq:load_gen}
\end{equation}
then $\rho < 1$ is a necessary stability condition.
\end{prop}
\indent 
\indent
\begin{proof} Denote by $\rho_{j}$ the load in area $A_j$ associated with both Carriers 1 and 2; $\rho_{j}$ is defined as the traffic intensity offered by users in area $A_j$ divided by the total capacity $C_{1,j} + C_{2,j}$ in $A_j$, namely, 
\begin{equation}
\rho_j = \frac{( \alpha_j + \beta_j ) \sigma}{C_{1,j} + C_{2,j}}.
\label{defrhoj}
\end{equation}
According to \cite{bon03}, a necessary stability condition for such a work-conserving multiclass systems is that the total offered load on both Carriers 1 and 2 must be less than 1, that is, $\sum_{j=1}^J \rho_{j} < 1$ which, in view of definitions (\ref{eq:cap}), (\ref{eq:load_gen}) and (\ref{defrhoj}), is equivalent to inequality $\rho < 1$, as claimed. 
\end{proof}

A fluid approach similar to that of Proposition 1 proves, however, difficult to extend to the multi-area case to justify the sufficiency of the stability condition $\rho<1$.   
The following simple argument can, nevertheless, be proposed in the favor of this {\it conjecture}. First define the  simple {\it Bernoulli scheduling policy} as follows: a flow arriving in area $A_j$ is scheduled on Carrier 1 with probability $C_{1,j}/ (C_{1,j} + C_{2,j})$ and on Carrier 2 with probability $C_{2,j}/ (C_{1,j} + C_{2,j}) $ (here both SC and DC traffic is scheduled according to the Bernoulli policy). For the Bernoulli policy, both carriers behave independently as a multiclass PS queue; a necessary and sufficient stability condition for Carrier 1  
under the Bernoulli policy is then (\cite{Borst03}, Proposition 3.1)
\begin{equation}
\sum_{j=1}^J \rho_{1,j} < 1
\label{eq:stabBer}
\end{equation}
where $\rho_{1,j}$ is the load on Carrier 1 induced by users from area $A_j$, that is,
\begin{equation}
\rho_{1,j} = \left [ (\alpha_j+\beta_j)\frac{C_{1,j}}{C_{1,j} + C_{2,j}} \right ] \times \frac{\sigma}{C_{1,j}} = \rho_j
\label{eq:rhoj}
\end{equation}
with $\rho_j$ given in (\ref{defrhoj}) (we similarly show that the load on Carrier 2 induced by users from area $A_j$ is $\rho_{2,j} = \rho_j$). In view of (\ref{eq:load_gen}) and (\ref{eq:rhoj}), condition (\ref{eq:stabBer}) equivalently reads $\rho < 1$. Considering the stability region for a given policy as that associated with a positive throughput, the above discussion shows that the throughput $\gamma^B$ for the Bernoulli policy is strictly positive in the capacity region defined by $\rho < 1$. Given that JFQ and VB scheduling policies both take into account the state of the system, they are expected to perform better than the blind Bernoulli scheduling in terms of throughput, we have $\gamma \ge \gamma^B > 0$ as soon as $\rho<1$. We could thus infer that the multiclass system implementing JFQ and VB scheduling is also stable in the region defined by  $\rho<1$. 

\subsection{Throughput performance}
\label{sec:SD}
If stability condition $\rho < 1$ is fulfilled, the system has a stationary distribution. Let then
$$
\Pi(\mathbf{n}) = \lim_{t \uparrow +\infty} \mathbb{P}(\mathbf{n}(t) = \mathbf{n}), \; \; \; \mathbf{n} \in \mathbb{N}^{3J},
$$
define the stationary distribution of the Markov process $\mathbf{n}(t)$, $t \geq 0$. 
That stationary distribution can be computed by writing the associated balance equations which are not given here for the sake of brevity. Solving these equations enables us to determine the stationary distribution $\Pi$ which in turn allows us to derive various performance indicators. In particular, we are interested in deriving the mean flow throughput defined as the ratio of the mean flow size to the mean flow duration. Using Little's law, the mean flow throughput $\gamma_{SC,j}$ and $\gamma_{DC,j}$ for SC and DC flows arrived  in area $A_j$ can be  expressed by
\begin{equation}
\gamma_{SC,j} = \frac{\alpha_j \sigma}{\mathbb{E}(n_{1,j} + n_{2,j})}, \; \; \; \; \; \gamma_{DC,j} = \frac{\beta_j \sigma}{\mathbb{E}(m_j)}.
\label{Gammas}
\end{equation}

Assume for instance that there are no SC users in the system. In view of (\ref{defD1D2}), rates are given by $D_{1,j}(m) = C_{1,j}/m$ and $D_{2,j}(m) = C_{2,j}/m$ when the system contains $m$ DC flows, and thus all DC clients in area $A_j$ are served with a total rate $C_{1,j}+C_{2,j}$. The system then corresponds to a multiclass PS queue where each area $A_j$ defines a different service class; the number of clients in each area $A_j$ can thus be written as
\begin{equation}
{E}(m_j)=\frac{\rho_j}{1-\rho}
\label{eq:noClients}
\end{equation} 
with $\rho=\sum_{j=1}^J \rho_j$ and $\rho_j=\beta_j\sigma/(C_{1,j}+C_{2,j})$ is the load due to DC users in area $A_j$. From (\ref{Gammas}) and (\ref{eq:noClients}), we then obtain
\begin{equation}
\gamma_{DC,j} = (C_{1,j}+C_{2,j})(1-\rho).
\label{labelDConly}
\end{equation}

In view of (\ref{labelDConly}),it is important to note that when only DC users are present in the system, the VB scheduling policy achieves ideal load balancing between the two carriers. Indeed, the throughput obtained by the DC flows under VB is equal to the throughput obtained for a single carrier of capacity $C_{1,j}+C_{2,j}$ shared according to the PS policy. 




\section{Same capacity carriers}


In this section, we analyze the case where both carriers have equal capacities. This corresponds to the Dual Carrier feature of HSDPA in Release 8 where two identical carriers of 5 MHZ, both in the 2.1 GHZ band, are aggregated. For the sake of simplicity, we consider the case of a single area, i.e. $J=1$, and drop index $j$ in the notation. The general case where $J>1$ and $C_{1,j} \ne C_{2,j}$, $1 \le j \le J$, is considered in Section \ref{sec:res}. We thus set $C_1=C_2=C$. 
Let $\phi$ denote the proportion of SC traffic out of the total traffic, so that  $\alpha=\phi\Lambda$ and $\beta = (1-\phi)\Lambda$. We set $C=1$ such that results are normalized to $C$.   

The throughput performance of DC users when only DC devices are present in the cell is simply given by (\ref{labelDConly}) for $C_1=C_2=C$, that is,

\begin{equation}
  \gamma_{DC} = 2C(1-\rho).
\label{eq:gamma_DC_sameCap}
\end{equation}

\subsection{SC users only}
\label{sec:SConly}

SC users are distributed among the two queues according to the JFQ discipline. When the two carriers have equal capacities, i.e., $C_{1} = C_{2}$ for the JFQ policy coincides with to the well-known \textit{Join the Shortest Queue} (JSQ) policy. Indeed, the shortest queue in terms of the number of active users is then also the queue which yields the smallest completion time,  which coincides with JSQ when the two carriers have the same capacity. Note by symmetry that the average number of customers in both queues is equal, i.e., $\mathbb{E}(n_1)=\mathbb{E}(n_2)=n(\rho)$. An analytical expression for $n(\rho)$ for two parallel queues of equal capacity ruled by the JSQ policy is derived in \cite{fla77}. Based on this result and using Little's law, the throughput of SC users under the JSQ policy can be written as
\begin{equation}
 \gamma_{SC}=\frac{\alpha\sigma}{n(\rho)}=\frac{C\rho}{n(\rho)}
 \label{eq:gamma_JSQ}
\end{equation}
with $\rho$ defined in (\ref{eq:rhoJ1}) and $n(\rho)$ igiven in  \cite{fla77}.

\begin{figure}[t]
\begin{center}
\includegraphics[width=7.9cm,trim=2cm 8cm 2.5cm 8cm, clip]{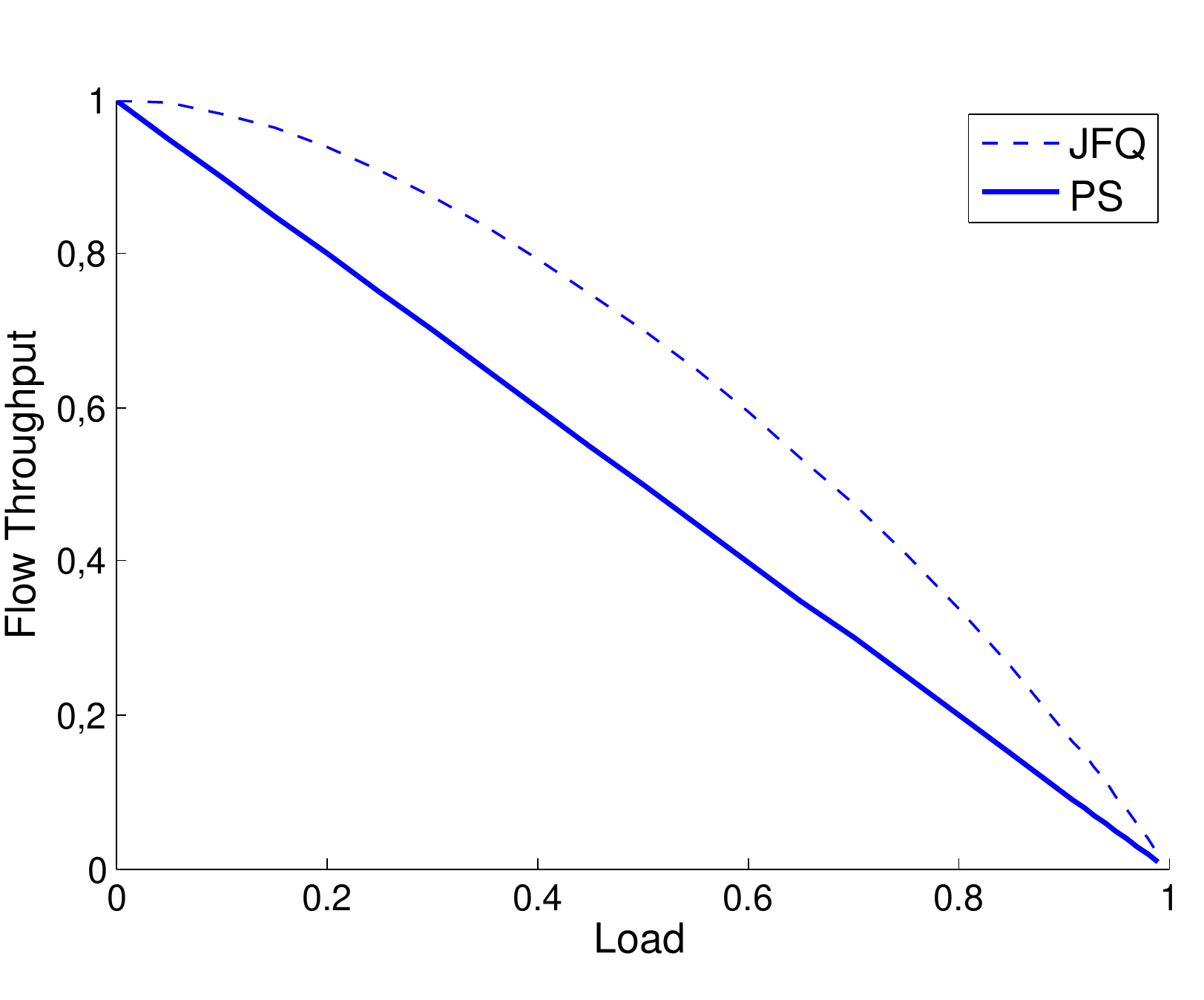}
\caption{\label{fig:JSQ} Mean throughput of SC users in terms of load: JSQ scheduling over 2 carriers of capacity is $C=1$ vs. PS scheduling over one carrier of capacity $C=1$.}
\end{center}
\end{figure}

Figure \ref{fig:JSQ} gives the mean flow throughput expressed in (\ref{eq:gamma_JSQ})  and compares it to the throughput obtained if the carrier is shared according to the PS policy. At very low network load,  PS and JSQ attain the same throughput. Indeed, when  a single SC customer is present, it is naturally provided with the maximal throughput $C$ of a single server, that is, $\gamma_{SC} = C$ for $\rho = 0$. 
As the load increases, however, the joint scheduling over 2 carriers with capacity $C$ each yields considerably higher throughput than that provided by a single PS queue of capacity $C$. 
Indeed, using results in  \cite{fla77}, it can be shown from (\ref{eq:gamma_JSQ}) that $ \gamma_{SC} \sim 2C(1-\rho$), i.e., $\gamma_{SC}$ is asymptotic to $\gamma_{DC}$ the throughput of DC users given in (\ref{eq:gamma_DC_sameCap}). In other words, at high network load, JSQ allows SC users to efficiently utilise both carriers and to attain a throughput which is equivalent to that of DC users.  

\subsection{Joint performance of SC and DC users}

We now consider the practically interesting case in which both SC and DC users are present in the cell.  Figure \ref{fig:SCDC1} represents the flow throughput as a function of the network load for different values of $\phi$. The results are obtained by numerically solving the balance equations of the Markov process $\mathbf{n}(t)$ presented in Section \ref{sec:SD}.  Note that truncation of the state space  is necessary in order to numerically solve the balance equations. This engenders less accuracy at high network load, say $\rho>0.9$. 

We note that the performance of DC users is slightly improved when $50\%$ of the traffic originates from SC users. This is because each SC user is constrained to using a single carrier which is favourable to DC users; conversely, the performance of SC users is slightly degraded by the presence of DC users. The results of Figure \ref{fig:SCDC1}, however, indicate that mixing both SC and DC traffic does not significantly impact the throughput performance of neither SC nor DC users. 
\begin{figure}[ht]
\begin{center}
\includegraphics[width=9.3cm,trim=3cm 9cm 2.5cm 9cm, clip]{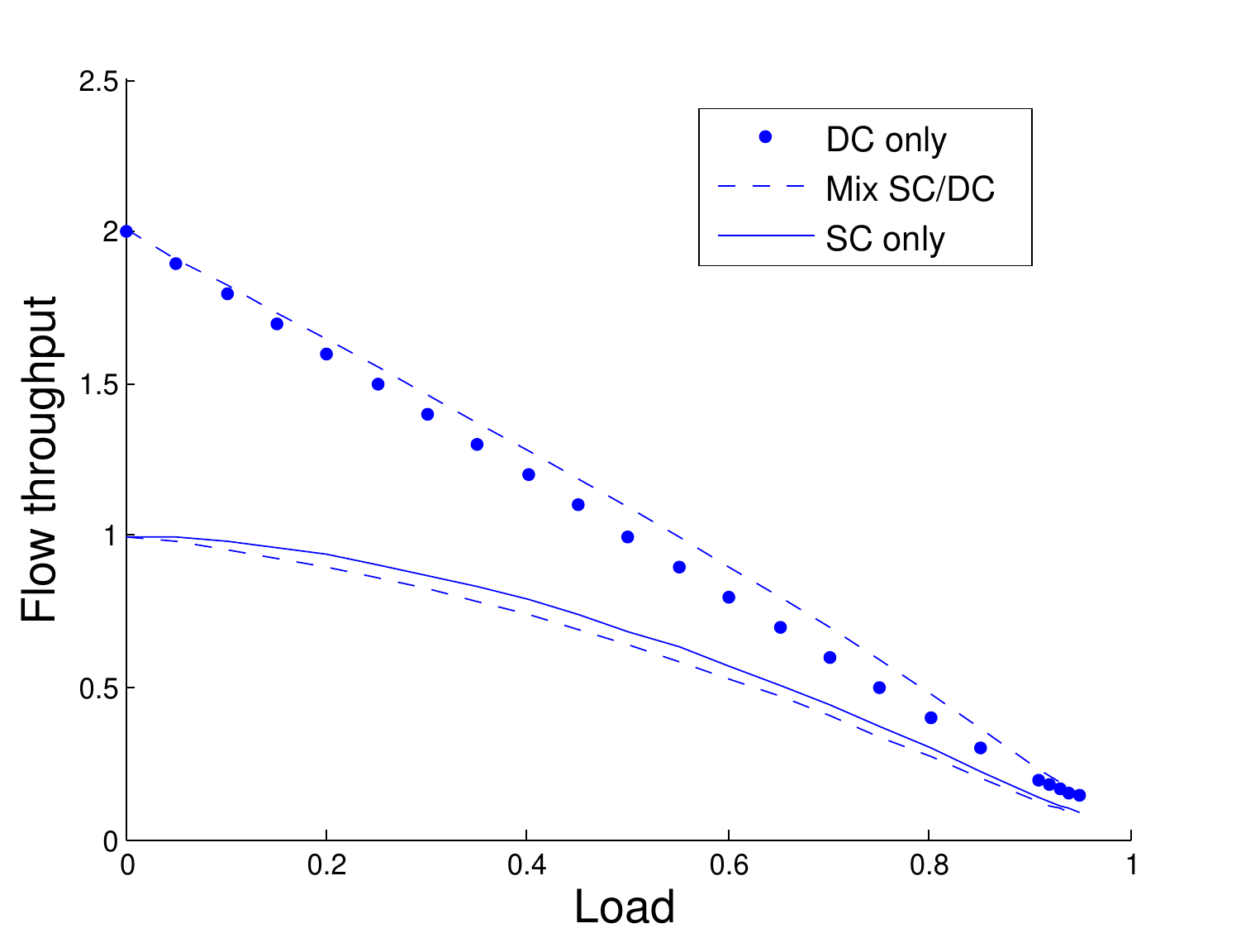}
\caption{\label{fig:SCDC1} \textit{Mean throughput of SC and DC users in terms of network load for different values of $\phi$: $\phi=0$ (only DC traffic), $\phi=0.5$ (mix of SC and DC traffic) and  $\phi=1$ (only SC traffic) and for $C=1$.}}
\end{center}
\end{figure}


\section{Different capacity carriers}
\label{sec:res}


We now consider the more general case of two carriers having different capacities, i.e. $C_1 \ne C_2$. This corresponds to  Dual Band HSDPA and in LTE-A with carriers of different sizes.  At the end of this section, we also consider the case of having two different areas in the cell. In the first part of the section, we normalize results with respect to $C_1$, i.e.,  $C_1=1$ and $C_2$ is either $C_2=1.3$ or $C_2=2$ in order to model DB HSDPA and LTE-A networks. The choice of these values is motivated as follows. DB HSDPA aggregates two carriers of 5 MHZ each, but the one in the 900 MHZ band has a slightly larger capacity due to favourable propagation conditions. On the other hand, LTE-A allows to aggregate carriers of different capacities; for instance,  a carrier of 20 MHZ  can be aggregated with a carrier of 10 MHZ within the same frequency band. A numerical application for HSDPA networks is presented in Section \ref{sec:HSDPA}.  

The performance of DC users is given by (\ref{labelDConly}) for $J=1$. We subsequently consider the case in which there are only SC users in the system and the case of mixed SC and DC users. 


\subsection{SC users only}
\label{sec:SConlydiffCapa}
We  compare the proposed JFQ discipline to 
the well-known JSQ discipline in the case in which only SC users are present in the cell, i.e., $\phi=1$. 


Figure \ref{fig:JFQ} represents the mean throughput of SC flows under the two scheduling disciplines for $C_1=1$ and $C_2=2$. The JFQ discipline clearly outperforms JSQ. This is because JFQ takes scheduling decisions based not only on the queue size but also on the capacity of each queue. 
Based on the numerical evaluations, the throughput of JFQ can be approximated by 

\begin{equation}
\gamma_{SC}\approx C_2(1-\rho),
\label{eq:SCapprox}
\end{equation}
which corresponds to the throughput of a PS server of capacity $C_2$. As $C_2>C_1$, the JFQ policy therefore yields a performance similar to that obtained when the users are always assigned to the carrier with the largest capacity.

\begin{figure}[ht]
\begin{center}
\includegraphics[width=8.3cm, trim=2cm 8cm 2.5cm 8cm, clip]{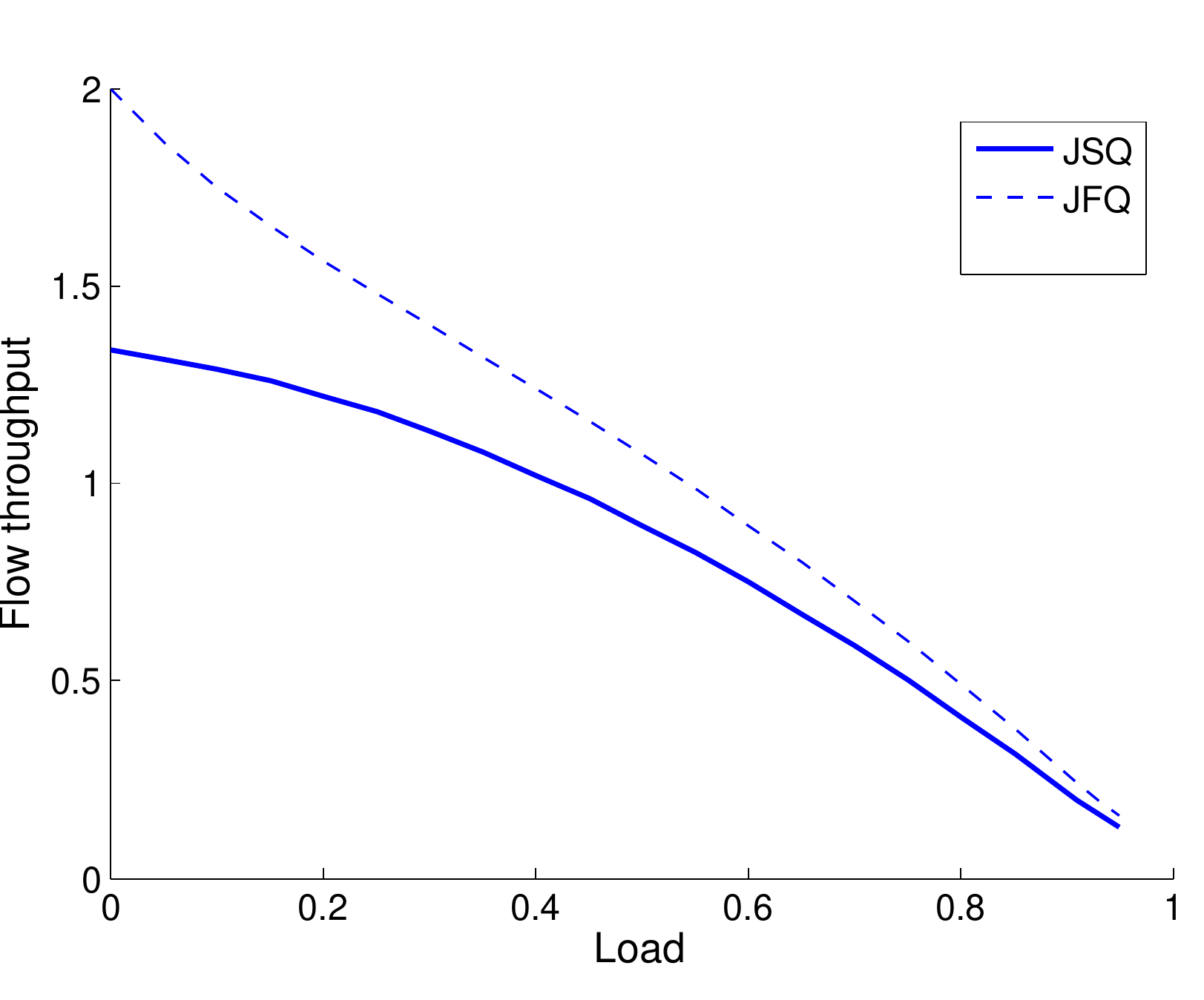}
\caption{\label{fig:JFQ} \textit{Mean throughput of SC users under JSQ and JFQ scheduling for $C_1=1$ and $C_2=2$.}}
\end{center}
\end{figure}



\subsection{Joint performance of SC and DC users}
\label{sec:impactP}
Assume now that both SC and DC users are present in the cell. Throughput performance is given in Figure \ref{fig:SCDC_diff}  for different values of $\phi$ and for $C_1=1$, $C_2=1.3$. As in the case of equal capacity servers, we notice a slight improvement of DC performance when SC traffic is present in the system.  Once more, the impact of $\phi$ is not considerable. 
In other words, since performance is quasi-independent of the mix of SC and DC traffic, we can study the performance of each class independently by considering the system with only SC or only DC traffic. 

\begin{figure}[t]
\begin{center}
\includegraphics[width=9.3cm, trim=2cm 8.4cm 2.5cm 9cm, clip]{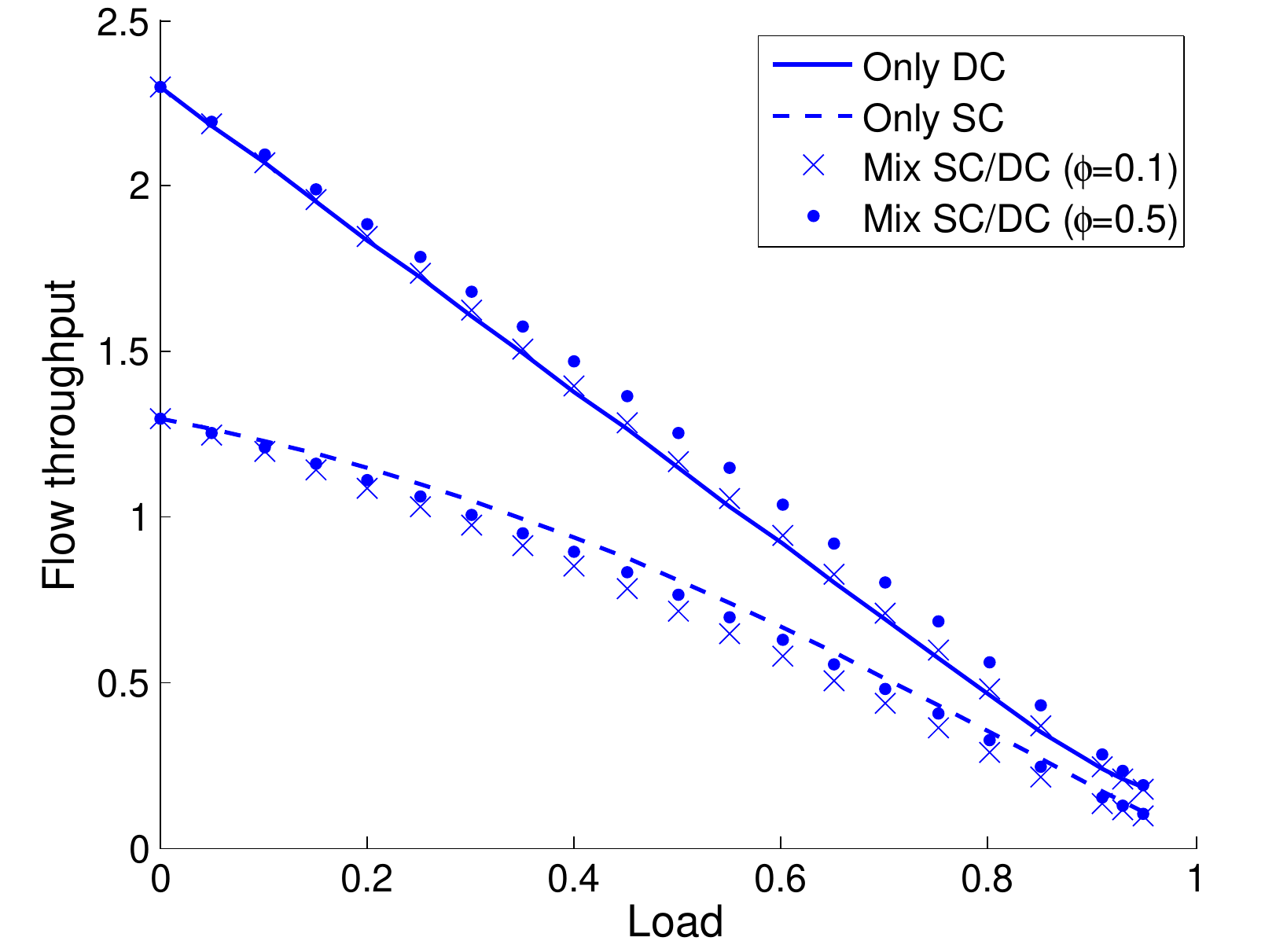}
\caption{\label{fig:SCDC_diff} \textit{Mean throughput of SC and DC users in terms of network load for different values of $\phi$: $\phi=0$ (only DC traffic), $\phi=0.1$, $\phi=0.5$ (mix of SC and DC traffic) and  $\phi=1$ (only SC traffic) and for $C=1$.}}
\end{center}
\end{figure}

Figure \ref{fig:3Loads} shows the impact of the percentage of DC traffic on the  {\it average throughput} of the cell for different values of the load. The average troughput in area $A_j$, $\overline{\gamma}_j$, is defined by
\begin{equation}
\overline{\gamma}_j = \phi \cdot \gamma_{SC,j}+(1-\phi) \cdot \gamma_{DC,j}.
\label{eq:gamma_tot}
\end{equation}
Note that at any network load the total throughput increases as the percentage of DC traffic increases. The increase is more significant at lower network load. This can be explained as follows. We have seen that $\gamma_{SC}$ and $\gamma_{DC }$ are quasi-insensitive to variations of $\phi$. In view of (\ref{eq:gamma_tot}), $\overline{\gamma}$ varies quasi-linearly with $\phi$ with slope $\gamma_{SC}-\gamma_{DC}$. As shown in Figure \ref{fig:SCDC_diff}, this slope is maximal at low load. 

\begin{figure}[t]
\begin{center}
\includegraphics[width=7.9cm, trim=2cm 8cm 2.5cm 8cm, clip]{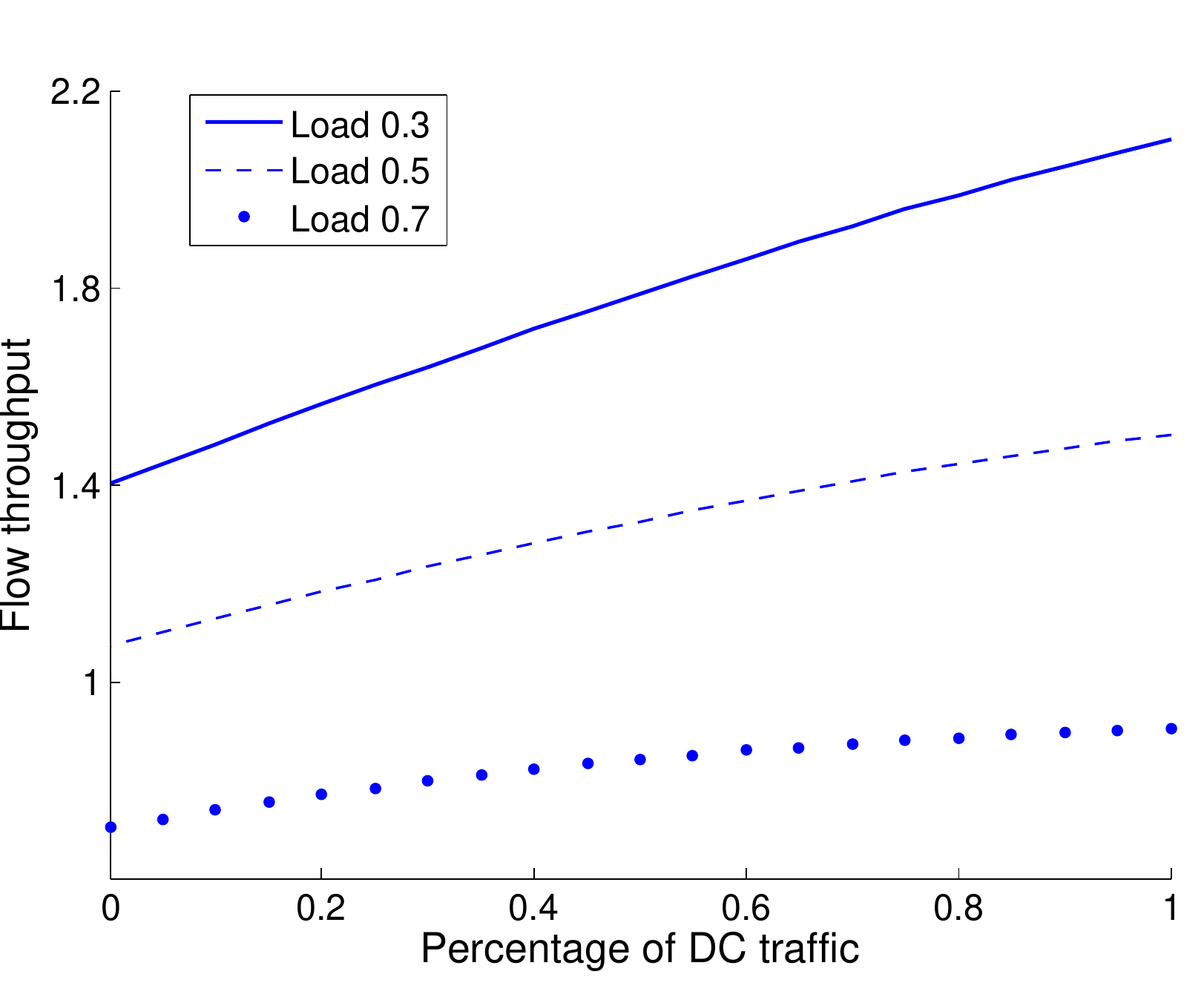}
\caption{\label{fig:3Loads} \textit{Mean throughput of SC and DC users for different values loads and for $C_1=1$, $C_2=2$.}}
\end{center}
\end{figure}



\subsection{Application to HSDPA and LTE networks}
\label{sec:HSDPA}
We are now interested in applying our model so as to estimate the traffic capacity of Dual Cell HSDPA and Dual Band HSDPA. We consider a cell having two distinct areas, one corresponding to transmission conditions in the cell center and the other corresponding to transmission conditions at the cell edge. Our aim is to determine the traffic that can be sustained by the cell such that SC users in the cell edge attain a certain target throughput. 

We consider a cell with a total radius of $R=600$ m and consider that users are uniformly distributed over the cell. In \cite{900MHZ}, the capacities of DC and DB HSDPA are obtained by means of a static system level simulator. 
According to \cite{900MHZ}, capacities of DC HSDPA correspond to $C_{1,1} = C_{2,1}=10$ Mbit/s, while DB HSDPA has  $C_{1,1} =10$ Mbit/s and $C_{2,1}=14$ Mbit/s, for carriers in the 2.1 GHZ and the 900 MHZ band, respectively.  For LTE, we consider $C_{1,1} =150$ Mbit/s and $C_{2,1}=70$ Mbit/s. In all cases, carrier capacities in the cell edge are about ten times less than those at the cell center and 50\% of the users are in the cell center.

Table \ref{tab:2zones_thr} shows the maximum traffic intensity that can be sustained by a HSDPA or LTE cell such that the target average throughput at the cell edge is 1 Mbit/s and 10 Mbit/s, respectively. As for the total average throughput (cf. \S \ref{sec:impactP}), the traffic intensity increases as the percentage of DC traffic increases, for both HSDPA and LTE. As expected, the traffic capacity of DB HSDPA is superior to that of DC HSDPA. 

In the case of LTE-A, in which one carrier may have a capacity largely superior to the other, we have seen in Section \ref{sec:SConlydiffCapa} that the throughput of SC users can be approximated by (\ref{eq:SCapprox}). Taking into account the system's insensitivity to $\phi$ and the throughput of DC users given by (\ref{labelDConly}), we are able derive the traffic intensity $\Theta = (\alpha+\beta)\sigma$. In view of definitions (\ref{eq:load_gen}) and (\ref{eq:gamma_tot}) applied to the edge area $j = J$, $\Theta$ can consequently be approximated by
$$
\Theta \approx \overline{C} \left( 1- \frac{\overline{\gamma}_J}{(1-\phi)C_{1,J}+C_{2,J}}\right).
$$

We verify that the above approximation is consitent with the values given in Table \ref{tab:2zones_thr} for the LTE network. 

\begin{table}
\begin{center}
  \caption{\label{tab:2zones_thr} \textit{Traffic intensity for different percentages of SC traffic ($\phi$) when the target average throughput at cell edge  is $\overline\gamma_J=$1 Mbit/s for HSDPA, and $\overline\gamma_J=$10 Mbit/s for LTE}}
  \begin{tabular}{ | c | c | c | c | c | }
    \hline
     & $\phi=1$ & $\phi=0.8$ & $\phi=0.5$ & $\phi=0.2$ \\ \hline
      DB HSDPA (Mbit/s) &  1.75 & 2.11 & 2.38 &  2.65  \\ \hline
      DC HSDPA (Mbit/s) &  - & 1.08  &  1.48  & 1.73 \\ \hline \hline
      LTE (Mbit/s)& 12.8 & 16 & 19.6  & 24.8 \\ \hline
	
  \end{tabular}
\end{center}
\end{table}




\section{Conclusion}
\label{sec:concl}

This paper addresses the performance of load balancing and  carrier aggregation in HSPA+ and LTE-A networks. In CA networks, a joint scheduler is used to manage two or more carriers. Load balancing schemes are needed in order to evenly distribute the incoming traffic.  We have proposed two such schemes: JFQ which allows to distribute the traffic of non-CA users, and VB which balances the traffic of CA users. By means of mathematical modeling, we have shown that both schemes achieve quasi-ideal load balancing over the carriers. Indeed, when only DC users are present in the system,  the throughput obtained by DC flows under VB is equal to the throughput obtained for a single carrier of capacity $C_{1,j}+C_{2,j}$ shared according to the PS policy; VB thus maximizes the utilization of the two carriers. SC users, on the other hand, are constrained to using only one of the two carriers. We have seen that JFQ favors the usage of the highest capacity carrier, thus maximizing carrier utilization. 

The proposed model also allows us to gain insight into the performance of multi-carrier networks in which traffic is generated by both CA and non-CA users. We have shown that the throughput of both  CA and non-CA users is practically insensitive to the percentage of CA users.  
Consequently, we can evaluate the performance of each class independently by considering the system with only SC or only DC traffic.
In future work, we intend to intend to extend the developed models to a larger number of aggregated carriers. 

\bibliographystyle{IEEEtran}
\bibliography{IEEEabrv,biblio}


\end{document}